\def\year{2019}\relax
\documentclass[letterpaper]{article} 
\usepackage{aaai19}  
\usepackage{times}  
\usepackage{helvet}  
\usepackage{courier}  
\usepackage{url}  
\usepackage{graphicx}  
\usepackage{amsmath,amscd,amsbsy,amssymb,latexsym,url,bm,amsthm,mathtools}
\usepackage{stfloats} 
\usepackage[ruled,lined]{algorithm2e}
\newtheorem{theorem}{Theorem}
\newtheorem{definition}{Definition}
\newtheorem{proposition}{Proposition}
\newtheorem{corollary}{Corollary}[theorem]
\newtheorem{lemma}{Lemma}

\begin{document}
%
\title{Cooperation Enforcement and Collusion Resistance \\in Repeated Public Goods Games}
\author{Kai Li\\
Shanghai Jiao Tong University\\
kai.li@sjtu.edu.cn\\
\And
Dong Hao\thanks{Corresponding author.}\\
University of Electronic Science and Technology of China\\
haodong@uestc.edu.cn
}
\maketitle

\begin{abstract}
Enforcing cooperation among substantial agents is one of the main objectives for multi-agent systems. However, due to the existence of inherent social dilemmas in many scenarios, the free-rider problem may arise during agents' long-run interactions and things become even severer when self-interested agents work in collusion with each other to get extra benefits. It is commonly accepted that in such social dilemmas, there exists no simple strategy for an agent whereby she can simultaneously manipulate on the utility of each of her opponents and further promote mutual cooperation among all agents. Here, we show that such strategies do exist. Under the conventional repeated public goods game, we novelly identify them and find that, when confronted with such strategies, a single opponent can maximize his utility only via global cooperation and any colluding alliance cannot get the upper hand. Since a full cooperation is individually optimal for any single opponent, a stable cooperation among all players can be achieved. Moreover, we experimentally show that these strategies can still promote cooperation even when the opponents are both self-learning and collusive.
\end{abstract}

\section{Introduction}

The emergence of substantial real-world multi-agent systems becomes an irreversible tendency in artificial intelligence.
Enforcing cooperation on the agents to achieve good outcomes is a long-standing challenge in multi-agent systems, especially when the agents are involved in a long-run interaction or when there are a large number of agents \cite{santos2018social,wooldridge2009introduction,kraus1997negotiation}.
One of the underlying reasons for such a challenge is that there is a social dilemma in many multi-agent systems \cite{panait2005cooperative,turner1993tragedy}.
 Although the agents can bilaterally share costs to give benefits to others and achieve a positive social equilibrium whereby each agent's utility is much improved, without a proper incentive or strategy design, the free-rider problem will occur and self-interested agents will settle to competitive equilibria that are most likely sub-optimal.
This well-known phenomenon is called the tragedy of the commons \cite{olson2009logic,hardin1968tragedy}.
It becomes even severer when some agents form a colluding group and tactically adjust the group strategy to get a higher average utility than that of the cooperative equilibrium \cite{telser2017competition,hilbe2014cooperation}.
Decades of research from different disciplines has sought to conquer the multi-agent social dilemma.
It remains equally critical today.

One of the most representative models for social dilemmas is the repeated Public Goods Games (PGG), which is a dynamic extension of the abstract public goods games \cite{ledyard1994public}.
A repeated PGG could consist of a series of stage games (either finitely or infinitely many).
In each stage $t$, the same group of $n$ agents privately decide how many of their tokens $e$ to put into a public pot.
At the end of each stage, the tokens in the pot are multiplied by a factor $r$, representing the `public good' and is equally distributed to all the agents. With $r$ smaller than the group size $n$, non-contributing free-riders gain more than contributors.
In literature, experimental economic studies suggest a common interval for $r$ is greater than $0.3n$ and less than $0.75n$ \cite{isaac1984divergent}.
Although this multi-player game system can yield a maximum total utility when all agents decide to fully contribute their tokens, the Nash equilibrium of a single stage PGG is no contribution to the public good for every agent.
In such a standard model of the repeated social dilemma, according to the Folk Theorem, the long-run interaction allows for the emergency of many equilibria within which the social optimum can be directed.
But the problem is, without proper strategy design for the repeated PGG, the individually rational agents will lead to an inferior outcome for all or some agents, than the decisions of `socially rational' agents.
The problem is how to select out good equilibria, which can be considered loosely equivalent to a social norm \cite{santos2017social,young2015evolution,epstein2001learning}.




In this work, we propose a new theory to solve the social dilemma under the framework of the repeated PGG.
We first set up the model of the $n$-player repeated PGG.
Via analysis of the joint Markov Decision Process (MDP) of this game, we introduce a compelling relation between a single player's strategy and the resulting state of the repeated game.
We then accordingly propose a theoretical method to offer one (or some) of the players a delicately designed Markov strategy, which we call cooperation enforcing.
Confronted with such a strategy, a single opponent can maximize his utility only via global cooperation by all players and any colluding alliance cannot get extra benefits.
Since a full cooperation is individually optimal for any single opponent, the social optimum, which means the stable cooperation among all players, can be achieved.
Moreover, one can find that the classic strategies including Grim Trigger and Win-Stay Loss-Shift all can be viewed as special cases of our proposed cooperation enforcing strategies.
Because any individual or group deviation from the social optimum will cause utility decrease, the proposed strategy is also collusion-resistant.
Although the proposed strategy is theoretically proved to have the above nice properties, we additionally show some numerical evidences.
In the simulation of the repeated PGG, the proposed strategy is run against an opponent group containing rational learning players.
The learning player is modeled to behave like human.
Such a simulation can help us understand the performance of the proposed strategy in the real world.

The significance of this newly proposed cooperation-enforcing strategy is multi-fold.
First, it identifies the individual player's power to unilaterally influence the payoff of each player in a group, which has not been discovered previously.
Based on this unilateral influence, the $n$-player mutual cooperation can be established simply by one player launching such a strategy, while any additional condition for cooperation such as initial condition, extra punishment and reward, social links, is not required. This makes the model simple but profound.
Second, it also allows us to understand how difficult it is for players to reach the social optimum and to sustain such an optimum.
It could be further utilized to establish a good social norm under different environments.
Third, it reduces the amount of optimization computation required by the players, since they do not have to search in their entire strategy space if they are aware of the fact that at least one opponent is using the cooperation enforcing strategy.
Last but not least, since under such strategies, all-player cooperation is the social optimum, any complicated multilateral deviation from it will cause damage to the average utility.
Thus these strategies naturally possess the property of resisting the collusion, which has been a very tough problem in multi-player games.
They provide us with a new possibility for solving collusion problems in more advanced game settings.


\section{Repeated Multiplayer Games}



Consider an infinitely repeated public goods game (PGG) among $n$ players with perfect monitoring.
In each stage, every player $i\in  \{1, \cdots, n\}$ decides whether to cooperate ($c$) by contributing her own endowment $e>0$ to the public pool, or to defect ($d$) by contributing nothing.
Without loss of generality, let $e=1$.
We denote player $i$'s action by $a_i \in A= \{c, d\}$.
After every player executes a certain action, there are totally $2^n$ possible outcomes for each stage.
Let ${\bm o}\in A^n$ denote the outcome vector.
To analyze the game from a group's perspective, we define a general representation of outcomes as follows.
\begin{definition}
	\label{def: Stage_Outcome_Representation}
A stage outcome can be represented as a tuple ${\bm o}=({\bm a}_{X}, {\bm a}_{Y})$, where ${\bm a}_{X}$ and ${\bm a}_{Y}$ are the action profiles from player groups $X$ and $Y$, respectively. $X\cap Y=\emptyset$ and $X \cup Y=\{1, \cdots, n\}$. Denote the number of cooperators and defectors in a group by functions $\#_c(\cdot)$ and $\#_d(\cdot)$, respectively.
\end{definition}
 For example, if $n=4$ and an outcome is ${\bm o}=ccdc$, then from a single player $i$'s perspective, ${\bm o}=(a_i, {\bm a}_{-i})$, where $a_i$ is the action of $i$ and ${\bm a}_{-i}$ is the action profile of other players. If $i=1$, then ${\bm a}_{-1}=cdc$ and $\#_c({\bm a}_{-1})=2$. From a joint perspective of two players $i$ and $j$, ${\bm o}=(a_i a_j, {\bm a}_{-ij})$, where ${\bm a}_{-ij}$ is the action profile excluding $i$ and $j$. If $i=1$ and $j=2$, then ${\bm a}_{-12}=dc$ and $\#_c({\bm a}_{-12})=1$. Besides, we define the $n$-player mutual cooperation and mutual defection as ${\bm o}=c^n$ and ${\bm o}=d^n$.

If the above game is infinitely repeated, then the strategy for each player is a mapping from any history to the current action $a$, which could be very complex.
We concentrate on the strategies depending only on what happened in the previous stage. These are called memory-one strategies and can be described as a vector of probabilities conditioning on the outcomes of the previous stage.
Intuitively, a memory-one strategy for the $n$-player PGG could be $2^n$-dimensional.
However, if the agents are symmetric players, then the only issue that matters for a player is how many opponents cooperate (i.e., $\#_c({\bm a}_{-i})$) and defect (i.e., $\#_d({\bm a}_{-i})$). Therefore, from a focal player $i$'s perspective, any stage outcome $\bm{o}\in A^n$ can be represented by a tuple $(a_i,k) \in A\times K$, where $a_i \in A= \{c, d\}$ is $i$'s action and $k =\#_c({\bm a}_{-i})\in K= \{0, 1, \cdots, n-1\}$ is the number of cooperating opponents. Then the dimension of the strategy space is reduced to $|A\times K|=2n$. Define the memory-one strategies for the repeated PGG as follows.

\begin{definition}
\label{def: Multi-Player Memory-One Strategies}
A memory-one strategy for player $i$ can be defined as a vector $\bm{p}=\left[p_{a_i,k}\right]$, where each $p_{a_i,k}= \mathbb{P} \{ c| {(a_i,k)} \}$ is the conditional probability for cooperation given the previous stage outcome $(a_i, k)$. More explicitly,
\begin{equation}
\begin{aligned}
\bm{p} =  \left (  p_{c, 0}, \cdots,  p_{c, k}, \cdots, p_{c, n-1}, \right. \\
\quad \left. p_{d, 0}, \cdots , p_{d, k} , \cdots, p_{d, n-1}  \right ).
\end{aligned}
\end{equation}

\end{definition}

Such a strategy vector contains $2n$ components, while the full outcome space is $2^n$-dimensional. This is because the strategy uses a same cooperating probability $p_{a_i,k}$ for some outcomes $(a_i,{\bm a}_{-i})$ that have the same number of cooperating opponents $k=\#_c({\bm a}_{-i})$.
Most of the well-known classic strategies can be written in the above form.
For instance, unconditional cooperation (\emph{ALLC}) can be given by $(1,1,\cdots, 1)$; unconditional defection (\emph{ALLD}) is written as $(0, 0, \cdots, 0)$; the \emph{Repeat} strategy,  which repeats the previous action, can be represented as $\bm{p}^{R}$ with $p_{c, k}=1$ and $p_{d, k}=0$ for all $k$.

At the end of each stage game, the stage outcome is perfectly observed by all players and the endowments in the current public pool are multiplied by a factor $r$, representing the `public good' and is equally distributed to all the agents.
Then the same stage game repeats. We call $r\over n$ the game's marginal per-capita rate of return (MPCR) \cite{isaac1984divergent}.
For every unit a player spends, the MPCR measures how much
the she gets back.
With $r$ smaller than the group size $n$, non-contributing free-riders gain more than contributors.

For repeated games, it is common to define payoffs as the average payoff that players receive over all stages.
To calculate these expected payoffs, we first define a focal player's exact stage game payoff values and payoff vectors.
\begin{definition}
	\label{def: Multi-Player Stage-Game Payoffs}
 Confronting with $k \in K$ cooperating opponents, player $i$'s stage payoff is $R_{a_i, k} = r(k+1)/n -1$ if $a_i=c$ and is $R_{a_i, k} = rk/n$ if $a_i=d$, where $r \in (1, n)$. Then the payoff vector for $i$ is $\bm{\pi}_i = (R_{a_i,k})_{({a_i,k})\in A\times K}$.
\end{definition}

Based on the payoff vectors, now one can derive the average payoffs. Denote by $\bm{v}(t)= \left (v_{\bm o}(t) \right)_{{\bm o} \in A^n}$ the probability distribution over the outcome ${\bm o}(t) \in A^n$ at the $t$-th stage game, where $v_{\bm o}(t)= \mathbb{P}\{{\bm o}(t)={\bm o}\}$.
Let $\bm{v}$ be a limit point of the sequence $\{\frac{1}{t} \sum_{m=1}^t \bm{v}(m) \}$ and we refer to it as \emph{limit distribution}.
Note that when the limit distribution is unrelated with the outcome of the initial stage game, then it is a unique \emph{stationary distribution} of the Markov chain formed by all players' memory-one strategies.
Player $i$'s expected payoff is an inner product of the limit distribution and her payoff vector, which is written as $\pi_i = \bm{\pi}_i \cdot \bm{v}$.

Both $\pi$ and $\bm v$ are associated with the Markov chain, which consists of all players' strategies.
But recent researches significantly reveal that in repeated games, there is an underlying relation between a single player's strategy $\bm p$ and the limit distribution of states $\bm v$ \cite{akin2012stable,hilbe2014cooperation}.
We write this result in a muti-player form as follows.
\begin{lemma}
	\label{lemma: Akin's Lemma}
Let $\bm{p}$ denote the focal player's memory-one strategy and let $\bm{p}^{R}$ denote the Repeat strategy, then for any limit distribution $\bm{v}$ (irrespective of the outcome of the initial stage), we have
	\begin{equation}\label{eq:Akin's Lemma}
	(\bm{p} - \bm{p}^{R}) \cdot \bm{v} = 0.
	\end{equation}
\end{lemma}
This general relation requires no additional information of the game. We call it Akin's Lemma, and it will be an important foundation for our following analysis.



\section{Cooperation Enforcing Strategy}


Traditionally, in repeated multiplayer games, one player cannot unilaterally manipulate the evolution of the game and it is even more difficult for a single player to enforce stable cooperation among all players.
However, based on the model in the previous section, now we will novelly analyze that one player can delicately design a Markov strategy whereby any single opponent can maximize his expected payoff only via global cooperation.
We call such strategies cooperation enforcing and formally define them as follows.


\begin{definition}[Cooperation Enforcing Strategy]
	\label{def: stable_cooperation strategies}
A memory-one strategy $\bm{p}$ for player $i$ is called cooperation enforcing if and only if the following conditions hold:\\
(1) Player $i$ cooperates in the first stage.\\
(2) $p_{c, n-1} =1$, i.e., player $i$ continues to cooperate if all the opponents cooperated in the previous stage.\\
(3) Either for all players $l \in \{1,2,\cdots, n\}, \ \pi_l = R_{c, n-1}$, or for any opponent $j \in \{1,2,\cdots, n\}\setminus\{i\}, \pi_j < R_{c, n-1}$.
\end{definition}
According to Definition \ref{def: stable_cooperation strategies}, if a focal player adopts a cooperation enforcing strategy, the payoffs of her opponents only have two possible forms, either less than or equal to the payoff of mutual cooperation.
As a result, the focal player not only controls the upper bound of each opponent's payoff, but also promote mutual cooperation.
Based on this definition, we can immediately obtain the following propositions.

\begin{proposition}
\label{proposition: strategy_p_2}
A memory-one strategy $\bm{p}$ is cooperation enforcing only if $p_{c, n-2} <1$.
\end{proposition}

\begin{proof}
Suppose that player $i$ uses a memory-one cooperation enforcing strategy $\bm{p}$ with $p_{c, n-2}=1$ while one of her opponents $j$ adopts \emph{ALLD} and all of the remaining players apply \emph{ALLC}.
In this case, the state is stationary in which $j$ defects and all his opponents cooperate.
As a result, the expected payoff of $j$ will be $R_{d, n-1} > R_{c, n-1}$, which contradicts the condition (3) in Definition \ref{def: stable_cooperation strategies}.
\end{proof}

\begin{proposition}
\label{proposition: r}
In a repeated public goods game, cooperation enforcing strategies exist only if $\frac{r}{n} > \frac{1}{2}$.
\end{proposition}

\begin{proof}
Assume that in a repeated public goods game with $\frac{r}{n} \le \frac{1}{2}$, $i$ adopts a memory-one strategy while $j$ uses \emph{ALLD} and the remaining players all use \emph{ALLC}.
In this case, no matter what $i$'s choice is in each round, $\pi_j$ is between $R_{d, n-2}$ and $R_{d, n-1}$.
From $\frac{r}{n} \le \frac{1}{2}$,  we have $R_{d, n-1} > R_{d, n-2}\ge R_{c, n-1}$.
As a result, $j$ could obtain an expected payoff $\pi_j > R_{c, n-1}$.
This contradicts point (3) in Definition \ref{def: stable_cooperation strategies}.
\end{proof}
The above two propositions indicate basic necessary conditions for the cooperation enforcing strategies, which could be convenient for verifying the suitable game structure and designing the strategies. Now we show that in a repeated PGG with perfect monitoring, if all players use certain cooperation enforcing strategies, they actually constitute an equilibrium. The equilibrium concept we use is Markov Perfect Equilibrium (MPE) \cite{maskin2001markov}, which is a refinement of subgame perfect equilibrium (SPE).
\begin{proposition}
	\label{proposition: MPE}
	If every player $i\in\{1, \cdots, n\}$ adopts a cooperation enforcing strategy $\bm{p}_i$, then the strategy profile $(\bm{p}_1, \bm{p}_2, \cdots, \bm{p}_n)$ is a Markov Perfect Equilibrium.
\end{proposition}
\begin{proof}
According to (1) and (2) in Definition \ref{def: stable_cooperation strategies}, every player chooses cooperating in the first stage and will continue cooperating if all the others also cooperate in the previous stage, thus the $n$-player mutual cooperation recursively sustains and every player obtains the same expected payoff $R_{c, n-1}$.
It is also known from (3) that if some players adopt cooperation enforcing strategies, the others' payoffs cannot exceed $R_{c, n-1}$. Cooperation is always the best response and any unilateral deviation induces a payoff loss.
Thus a profile of cooperation enforcing strategies is a Nash equilibrium. Since the game is infinitely repeated, the strategies actually constitute a sub-game perfect equilibrium.
The cooperation enforcing strategies we defined are Markov strategies, therefore, the profile of cooperation enforcing strategies is a Markov perfect equilibrium and generates a stable mutual cooperation.
\end{proof}

Definition \ref{def: stable_cooperation strategies} describes our goal for designing the cooperation enforcing strategies but the conditions are still not operational enough. Now we begin to do some reasoning and derivations, which will finally help us to obtain the detailed constraints for the cooperation enforcing strategies. We first show the following two lemmas as stepping stones.
\begin{lemma}
	\label{lemma: implication}
	The third condition in Definition \ref{def: stable_cooperation strategies} is equivalent to a logical implication:
	\begin{equation}
	\label{eq: implication}
	\exists j \ne i, \pi_j \ge R_{c, n-1} \Rightarrow v_{c^n}=1.
	\end{equation}
\end{lemma}
\begin{proof}
	The third condition in Definition \ref{def: stable_cooperation strategies} can be formally rewritten as a logical disjunction $(\forall j\ne i, \pi_j < R_{c, n-1}) \vee (\forall l, \ \pi_l = R_{c, n-1})$.
	According to basic logical theory, this disjunction is equivalent to an implication $\neg (\forall j\ne i, \pi_j < R_{c, n-1}) \Rightarrow (\forall l, \ \pi_l = R_{c, n-1}) $, which is further simplified as $\exists j \ne i, \pi_j \ge R_{c, n-1} \Rightarrow  \forall l, \ \pi_l = R_{c, n-1}$. Now we only need to prove the following equivalence relation
	\begin{equation}\label{eq: equivalence}
	\forall l, \ \pi_l = R_{c, n-1} \Leftrightarrow v_{c^n}=1.
	\end{equation}
Since the social optimum where every player obtains an expected payoff $R_{c, n-1}$ is reached only when mutual cooperation state is stable, i.e., $v_{c^n}=1$ and all the other elements in $\bm{v}$ are $0$, therefore eq. (\ref{eq: equivalence}) holds and this completes the proof.
\end{proof}



Lemma \ref{lemma: implication} is of great importance since it shows that the maximum payoff for any opponent $j$ is $\pi_j = R_{c, n-1}$ and is realized only when the $n$-player mutual cooperation is stable. Any type of unilateral deviation from cooperation will possibly incur payoff decrease, thus the temptation of free riding is suppressed. One can see that the cooperation enforcing strategy essentially sets up a payoff upper bound for each of the opponent. As long as the opponents are rational, in order to reach the payoff upper bound, they will sustain cooperation.

Now we know the basic mathematical constraints for $\bm v$. To derive the detailed constraints for the cooperation enforcing strategy $\bm p$, we need to further expand eq. (\ref{eq:Akin's Lemma}) in Akin's lemma and do some derivation and analysis.
To this end, we first need to analyze the game from a joint view of the focal player $i$ and a specific opponent $j$. Based on Definition \ref{def: Stage_Outcome_Representation}, we define the following marginal limit distribution.
\begin{definition}[Marginal Limit Distribution]
\label{def: marginal limit distribution}
Let $(a_ia_j, k)$ denote a stable outcome in which $i$ and $j$'s actions are $a_i$ and $a_j$, respectively, and there are $k=\#_c({\bm a}_{-ij}) \in M=\{0,1, \cdots, n-2\}$ cooperators among the other $n-2$ players. Denote by $u_{a_ia_j, k}$ the total probability of all such outcomes,
\begin{equation}
u_{a_ia_j, k} = \sum_{\substack{ \#_c(\bm{a}_{-ij})=k   }}v_{a_ia_j,\bm{a}_{-ij}} .
\end{equation}
Then $\bm{u} = (u_{a_ia_j, k})_{a_i, a_j \in A, k\in M}$ is the marginal limit distribution which is a $4(n-1)$-dimensional vector. Specifically, $u_{cc, n-2} = v_{c^n}$ and $u_{dd, 0} = v_{d^n}$.
\end{definition}

According to Definition \ref{def: marginal limit distribution}, now we can expand eq. (\ref{eq:Akin's Lemma}) and calculate the detailed relation between the game's limit distribution and a single player's strategy.
\begin{lemma}\label{lemma:rewrite_akin}
Given player $i$'s strategy $\bm{p}$ and the marginal limit distribution $\bm{u} = (u_{a_ia_j, k})_{a_i, a_k \in A, k\in M}$,
the formula in Akin's lemma for $\bm{p}$ can be rewritten as
\begin{equation}
\label{eq: adapted Akin's lemma}
\begin{aligned}
& (1-p_{c, n-2})u_{cd, n-2}\\
=& (-1 + p_{c, n-1})u_{cc, n-2} + (-1+p_{c, n-2})u_{cc, n-3}+ \cdots\\
&+ (-1+p_{c, k})(u_{cc, k-1} +  u_{cd, k})+p_{d, k}(u_{dc, k-1} + u_{dd, k})\\
&+ \cdots + p_{d, 0}u_{dd, 0}.
\end{aligned}
\end{equation}
Moreover, the difference between the expected payoff of $j$ and that of mutual cooperation is written as
\begin{equation}
\label{eq: gap between payoffs}
\begin{aligned}
& \pi_j - R_{c, n-1} \\
=& (R_{c, n-2}-R_{c, n-1})(u_{cc, n-3} + u_{dc, n-2})+\cdots \\
&+ (R_{c, k} - R_{c, n-1})(u_{cc, k-1} + u_{dc, k}) \\
&+ (R_{d, k} - R_{c, n-1})(u_{cd, k-1} + u_{dd, k}) + \cdots \\
&+ (R_{d, 0}-R_{c, n-1})u_{dd, 0}.
\end{aligned}
\end{equation}
\end{lemma}
\begin{proof}
	Eq. (\ref{eq: adapted Akin's lemma}) uses notations in Definition \ref{def: marginal limit distribution} and is directly derived from Akin's lemma.
	Since $\pi_j - R_{c, n-1} = (\bm{\pi}_j - R_{c, n-1}\bm{1})\cdot \bm{v}$. Merge similar terms, then we have eq. (\ref{eq: gap between payoffs}).
\end{proof}

Based on the above stepping stones, now we can propose a solution for the cooperation enforcing strategies.

\begin{theorem}
	\label{thm: main theorem}
	In the repeated public goods game with $r>\frac{n}{2}$, if a memory-one strategy $\bm{p}$ cooperates in the first stage and satisfies the following constraints:
	\begin{equation}
	\label{eq: main theorem}
	\left \{
	\begin{lgathered}
		p_{c, n-1} =1\\
		p_{c, n-2} <1\\
		p_{d, n-1} < \frac{(1-p_{c, n-2})(R_{c, n-1} - R_{c, n-2})}{R_{d, n-1} - R_{c, n-1}}\\
		p_{d, n-2} < \frac{(1-p_{c, n-2})(R_{c, n-1} - R_{d, n-2})}{R_{d, n-1} - R_{c, n-1}}\\
		\cdots \\
		p_{d, k} < \frac{(1-p_{c, n-2})(R_{c, n-1} - R_{d, k})}{R_{d, n-1} - R_{c, n-1}}\\
		\cdots \\
		p_{d, 0} < \frac{(1-p_{c, n-2})(R_{c, n-1} - R_{d, 0})}{R_{d, n-1} - R_{c, n-1}}
	\end{lgathered}
	\right. ,
	\end{equation}
	then $\bm{p}$ is a cooperation enforcing strategy.
\end{theorem}


\begin{proof}
Conditions (1) and (2) in Definition \ref{def: stable_cooperation strategies} are trivially satisfied. We only need to prove $\bm{p}$ meets condition (3).

Suppose that $i$ adopts $\bm{p}$ and $j$ is an opponent with objective $\pi_j \ge  R_{c, n-1}  $. Let $\bm{u}$ be the marginal limit distribution. By Proposition \ref{proposition: strategy_p_2}, we have $1-p_{c, n-2}>0$. Thus multiplying both sides by $1-p_{c, n-2}$ does not change an inequality. After that, introduce eq. (\ref{eq: adapted Akin's lemma}) into $(1-p_{c, n-2})(\pi_j - R_{c, n-1})\ge 0$, we have
\begin{equation}
\begin{aligned}
& (1-p_{c, n-2})(\pi_j - R_{c, n-1}) \\
=& b_{cc, n-3}u_{cc, n-3} + \cdots + b_{dd, 0} u_{dd, 0}\ge 0,
\end{aligned}
\end{equation}
where each $b_{a_i a_j, k}$ is the coefficient of $u_{a_i a_j, k}$.
According to Lemma \ref{lemma: implication}, we only need to prove the following implication is TRUE.
\begin{equation}
\label{eq: final implicaiton}
\begin{aligned}
& b_{cc, n-3}u_{cc, n-3}  +\cdots + b_{dd, 0} u_{dd, 0} \ge 0\\
\Rightarrow &  u_{cc, n-3}=\cdots=u_{dd, 0}=0.
\end{aligned}
\end{equation}
Since each element of the marginal limit distribution $u_{a_i a_j, k} \ge 0$, one sufficient condition for Implication (\ref{eq: final implicaiton}) to be TRUE is that all the coefficients $b_{a_i a_j, k} < 0$. Note in Lemma \ref{lemma:rewrite_akin} each $b_{a_i a_j, k}$ is a function of the corresponding $p $.

There are two corresponding constraints for $p_{c, k}$: $b_{cc, k-1} <0$ and $b_{cd, k}<0$.
These two inequalities require:
\begin{equation}
\label{eq: p_ck}
\left \{
\begin{lgathered}
 p_{c, k} < \frac{(1-p_{c, n-2})(R_{c, n-1} - R_{c, k})}{R_{d, n-1} - R_{c, n-1}} + 1\\
 p_{c, k} < \frac{(1-p_{c, n-2})(R_{c, n-1} - R_{d, k+1})}{R_{d, n-1} - R_{c, n-1}} + 1
\end{lgathered}
\right. .
\end{equation}
Since Theorem \ref{thm: main theorem} requests $r > \frac{n}{2}$, we have $R_{c, n-1} > R_{d, k+1} > R_{c, k}$ for all $k \in \{1,2,\cdots, n-3\}$.
As a result, the right sides of the inequalities in eq. (\ref{eq: p_ck}) are all greater than 1 and the constraints are all vanished.
As for $p_{c, 0}$, there is only one constraint $b_{cd, 0}<0$ and we can immediately verify that the constraint is satisfied.
The same conclusion can be derived for $p_{c, n-2}$.

Likewise, for $p_{d, k}$, we have the following constraint
\begin{align}
p_{d, k} <  \frac{(1-p_{c, n-2})(R_{c, n-1} - R_{d, k})}{R_{d, n-1} - R_{c, n-1}}.
\end{align}
As for $p_{d, n-1}$, it should satisfy the constraint $b_{d c, n-2}<0$, which results in
\begin{align}
p_{d, n-1} < \frac{(1-p_{c, n-2})(R_{c, n-1} - R_{c, n-2})}{R_{d, n-1} - R_{c, n-1}}.
\end{align}

Furthermore, because $R_{c, n-1} > R_{d, n-2}> \cdots > R_{d, 0}$ and $R_{c, n-1}> R_{c, n-2}$, the right sides of the above inequalities are all greater than 0, which makes these constraints enforceable. Collecting all of the inequalities above, we complete the proof that Implication (\ref{eq: final implicaiton}) is TRUE and $\bm p$ satisfies the third condition in Definition \ref{def: stable_cooperation strategies}.
\end{proof}

Another issue that deserves proper discussion is about what if multiple opponents deviate together and even make collusion.
The main aim of collusion is to achieve a level that the joint payoff can be higher than the equilibrium one.
How to conquer the failure of mutual cooperation caused by collusion is a long-existing tough problem in both game theory and multi-agent systems.
We will see in the following theorem, the cooperation enforcing strategy has a very good performance, even when some players collude to deviate from the mutual cooperation equilibrium.
\begin{theorem}
	\label{thm: collusion theorem}
In a repeated public goods game, if a cooperation enforcing strategy $\bm p$ exists, then it is collusion resistant.
\end{theorem}
\begin{proof}
If a cooperation enforcing strategy exists, then by Proposition \ref{proposition: r} we have $r>\frac{n}{2}$.
Moreover, by Lemma \ref{eq: implication} we know that under one player's cooperation enforcing strategy, the best response of other players is to cooperate, and the outcome ${\bm o}=c^n$ is an equilibrium.
Then if there is no collusion, the optimal payoff for each player is simply $R_{c,n-1}$.
Assume that there are $m\ge 2$ players collude to improve their average payoff, and other $n-m$ players stick to the equilibrium action $c$. If in the $m$ collusive players, $k$ of them cooperate and $m-k$ defect, then there are totally $k+n-m$ cooperators and $m-k$ defectors. Then the average expected payoff for the collusive players is:
\begin{equation}
\label{eq: collusive_payoff}
\tilde\pi= \frac{{kR_{c,k + n - m - 1}  + (m - k)R_{d,k + n - m} }}{m}.
\end{equation}
The difference between the collusive payoff and the equilibrium (cooperative) payoff is:
\begin{equation}
\label{eq: collusive_payoff_difference}
\tilde\pi-R_{c,n-1}= \frac{{(m - k)\left( {n - mr} \right)}}{{mn}}.
\end{equation}
This difference indicates collusion is beneficial only when $r<\frac{n}{m}$, but this contradicts with the necessary condition for $\bm p$.
Therefore, in a repeated PGG, if a cooperation enforcing strategy $\bm p$ exists, then it is collusion resistant.
\end{proof}

\begin{figure}[htbp]
	\centering
	\includegraphics[width=0.5\linewidth]{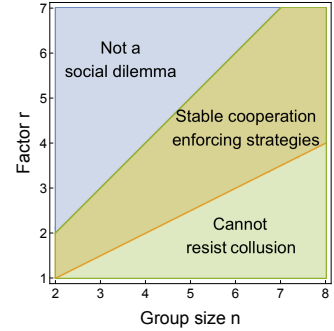}
	\caption{Schematic representation of parameter regions for the repeated public goods game.}
	\label{fig: r_n_relation}
\end{figure}

Figure \ref{fig: r_n_relation} shows in which parameter regions the cooperation enforcing strategies can exist, given that the public goods game forms a social dilemma (i.e., $1<r<n$).
When $\frac{n}{2}<r<n$, cooperation enforcing strategies always exist.
As $r$ approaches $\frac{n}{2}$ from the right side $r \to (\frac{n}{2})^+$, these strategies tend to behave like a grim trigger strategy.
When $r \le \frac{n}{2}$, one cannot promote mutual cooperation on her own and malicious collusive groups may gain the upper hand.

Specifically, based on the above results, we can identify multi-player versions for many well-known classic strategies in the two-player iterated prisoner's dilemma games, such as Grim Trigger (\emph{GT}) and Win-Stay Lose-Shift (\emph{WSLS}).
Based on Theorem \ref{thm: main theorem}, we can immediately obtain the following corollaries.

\begin{corollary}
(1) If $r>\frac{n}{2}$, an instantiation of the cooperation enforcing strategy $\bm p$ with $p_{c,n-1}=1$ and $p_{a_i, k}=0$ for all other states $(a_i, k)$ is a multi-player Grim Trigger (GT), and is collusion-resistant.\\
(2) If $r>\max \{\frac{n}{2},\frac{2n}{n+1}\}$, an instantiation of cooperation enforcing strategy $\bm p$ with $p_{c,n-1}=p_{d,n-1}=1$ and $p_{a_i, k}=0$ for all other states $(a_i, k)$ is a multi-player Win-Stay Lose-Shift (WSLS), and is collusion-resistant.
\end{corollary}
\begin{proof}
The proof for these two points is straightforward. Grim Trigger with $p_{c,n-1}=1$ and $p_{a_i, k}=0$ for all other $k$ satisfies all the constraints in Eqs. (\ref{eq: main theorem}). Win-Stay Lose-Shift also satisfies these constraints. Theorem \ref{thm: collusion theorem} indicates these two strategies are all collusion resistant.
\end{proof}

It is worth noting that, as far as we know, there are no widely accepted explicit definitions of multi-player versions of \emph{WSLS} and \emph{GT}. Take \emph{WSLS} as an example, what are ``win'' and ``lose'' is correlated with payoffs of all players, which is further affected by how many other players cooperate and defect. Thus, there could be more general definitions of the terms ``win'' and ``lose''. As shown in the above two corollaries, with the help of Theorem 1, we can at least identify multiple strategies which possess the good features that traditional \emph{WSLS} also has. These \emph{WSLS}-like strategies could be generalized versions of the traditional \emph{WSLS} into multi-player games.

Although Grim Trigger strategies can form a Markov perfect equilibrium, as this kind of strategy has no tolerance towards other players' mistakes, it may lead to poor results in the real world especially when there is noise and uncertainty.
\emph{WSLS} behaves more kindly than \emph{GT} and is more robust in noisy environments \cite{nowak1993strategy}.

To give an intuitive understanding of the effect of the cooperation enforcing strategies, now we use \emph{WSLS} to make a case study. In Figure \ref{fig: WSLS}, we show an example of how the focal player $x$ adopting a cooperation enforcing strategy constrains the expected payoffs of her opponents $y$ and $z$.
Player $x$ uses \emph{WSLS} strategy, whereas the strategy of $y$ and $z$ are randomly sampled from the space of memory-one strategies $10^5$ times.
The cube in the graph is the space of three players' payoff tuple $(\pi_x, \pi_y, \pi_z)$, with axises x, y and z representing the payoffs of player $x$, $y$ and $z$ respectively.
Each blue dot corresponds to a specific expected payoff tuple.
The pink plane $\pi_y=R_{c, n-1}$ on the left and the yellow one $\pi_z=R_{c, n-1}$ on the top illustrate the upper bounds of $y$ and $z$'s payoffs, respectively.
We can see the area of blue dots never crosses these two planes, which indicates that none of the opponents can obtain a payoff greater than that of mutual cooperation.
The social optimum $(R_{c,n-1}, R_{c, n-1}, R_{c, n-1})$ is a red point, which is the only point of intersection of the blue region and the two bound planes.

\begin{figure}[htbp!]
	\centering
	\includegraphics[width=0.9\linewidth]{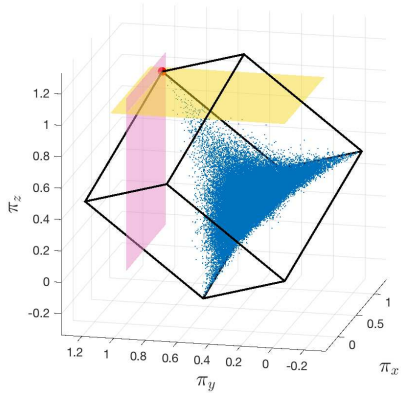}
	\caption{Illustration of a cooperation enforcing strategy in a 3-player repeated public goods game with $r=2$.
	The strategy of the focal player $x$ is \emph{WSLS} while the strategies of $y$ and $z$ are sampled randomly from the space of memory-one strategies.
	 Each blue dot corresponds to a feasible expected payoff tuple $(\pi_x, \pi_y, \pi_z)$ and the coordinate of the red point in the corner is $(R_{c, n-1}, R_{c, n-1}, R_{c,n-1})$.
	The pink (left) plane and the yellow (top) one are $\pi_y=R_{c, n-1}$ and $\pi_z = R_{c, n-1}$ respectively.}
	\label{fig: WSLS}
\end{figure}

\section{Cooperation Enforcement on Learning and Collusive Players}
In the previous sections, we theoretically identified the cooperation enforcing and collusion resistant strategies and found that a collection of such strategies constitutes a Markov perfect equilibrium.
In the real world multi-agent systems, if a player has no idea of the delicately designed strategies, being a learning player and gradually improving her own policy based on interaction history is a more realistic choice. In this section, we will show that, against the cooperation enforcing strategies, even the non-strategic learning opponent can gradually learn to cooperate.

From a learning player's point of view, a repeated game can be regarded as a sequential decision-making under uncertainty.
In each stage game, she makes a decision according to the history of interactions with her opponents.
In this paper, we only consider the history consisting of the previous stage game outcome.
The uncertainty comes from the inherent stochastic properties or evolution mechanisms of other players' policies.
For the learning player, her optimal plan in this MDP of the game can be solved by using dynamic programming (DP) based on the Bellman optimality equation \cite{mahadevan1996average}
\begin{equation}\label{eq: Q}
Q^*(\bm{o}, a) = \max_{a'\in A} \, \mathbb{E} \left [ R_{\bm{o}'} - R^* + Q^*(\bm{o}', a')\right].
\end{equation}
Here $Q^*(\bm{o}, a)$ is the action value function which estimates the score for selecting action $a$ after observing the previous stage game outcome $\bm{o}$.
After executing $a$, the learning player receives a result of the current stage $\bm{o}'$ and obtains a payoff $R_{\bm{o}'}$.
$R^*$ is the optimal average payoff.
There are many researches for solving the Bellman optimality equation and reinforcement learning is one which recently attracts much attention \cite{sutton2018reinforcement}. We use the average reward reinforcement learning approach \cite{gosavi2004reinforcement} and implement it in Algorithm \ref{alg: RL}.
\begin{algorithm}[htbp!]
	
	\label{alg: RL}
	\caption{A Learning Player's Strategy}
	
	Initialize a matrix: $Q(\bm{o}, a) \leftarrow 0$ for all $\bm{o}\in A^n, a\in A$ \;
	Initialize an estimate of the average payoff $\bar{R} \leftarrow 0$ \;
	Set outcome of the initial stage game $\bm{o}(0) \leftarrow c^n$\;
	Set the learning rate parameters $\alpha, \beta$\;
	
	\For{$t=1,2,\cdots$}{
		Take action $a$ with $\epsilon$-greedy policy based on $Q(\bm{o}(t-1), a)$\;
		Receive stage game outcome $\bm{o}(t)$ and payoff $R$\;
		$\delta \leftarrow R - \bar{R} + \max_{a'} Q(\bm{o}(t), a') - Q(\bm{o}(t-1), a)$\;
		$Q(\bm{o}(t-1), a) \leftarrow Q(\bm{o}(t-1), a) + \alpha \delta$\;
		\If{$Q(\bm{o}(t-1), a) = \max_{a'} Q(\bm{o}(t-1), a)$}{
			$\bar{R} \leftarrow (1-\beta)\bar{R} + \beta[ (t-1)\bar{R} + R]/t$;
	}
	}
	
\end{algorithm}

\begin{figure*}[tbp!]
	\centering
	\includegraphics[width=0.33\textwidth]{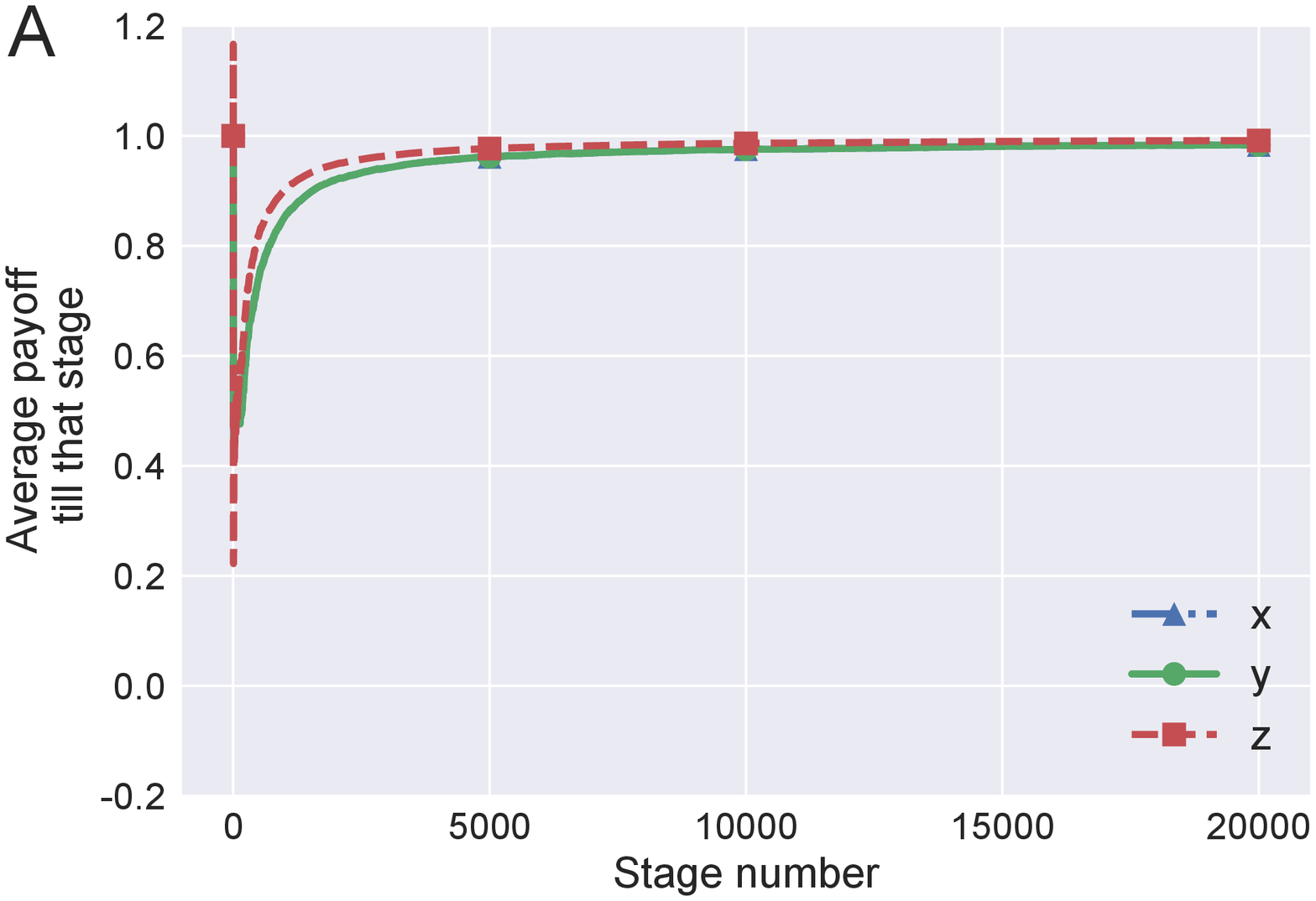}
	\includegraphics[width=0.33\textwidth]{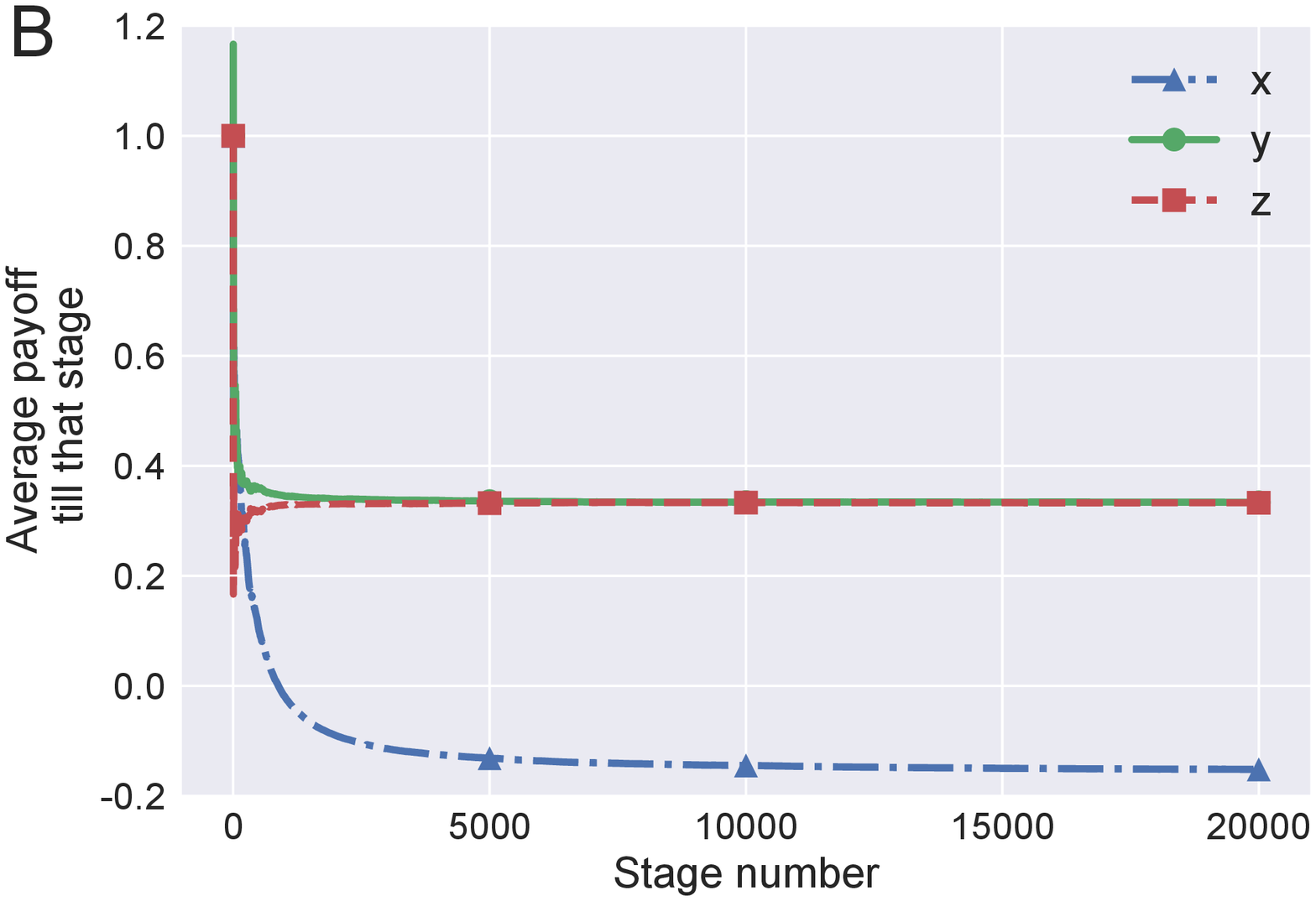}
	\includegraphics[width=0.33\textwidth]{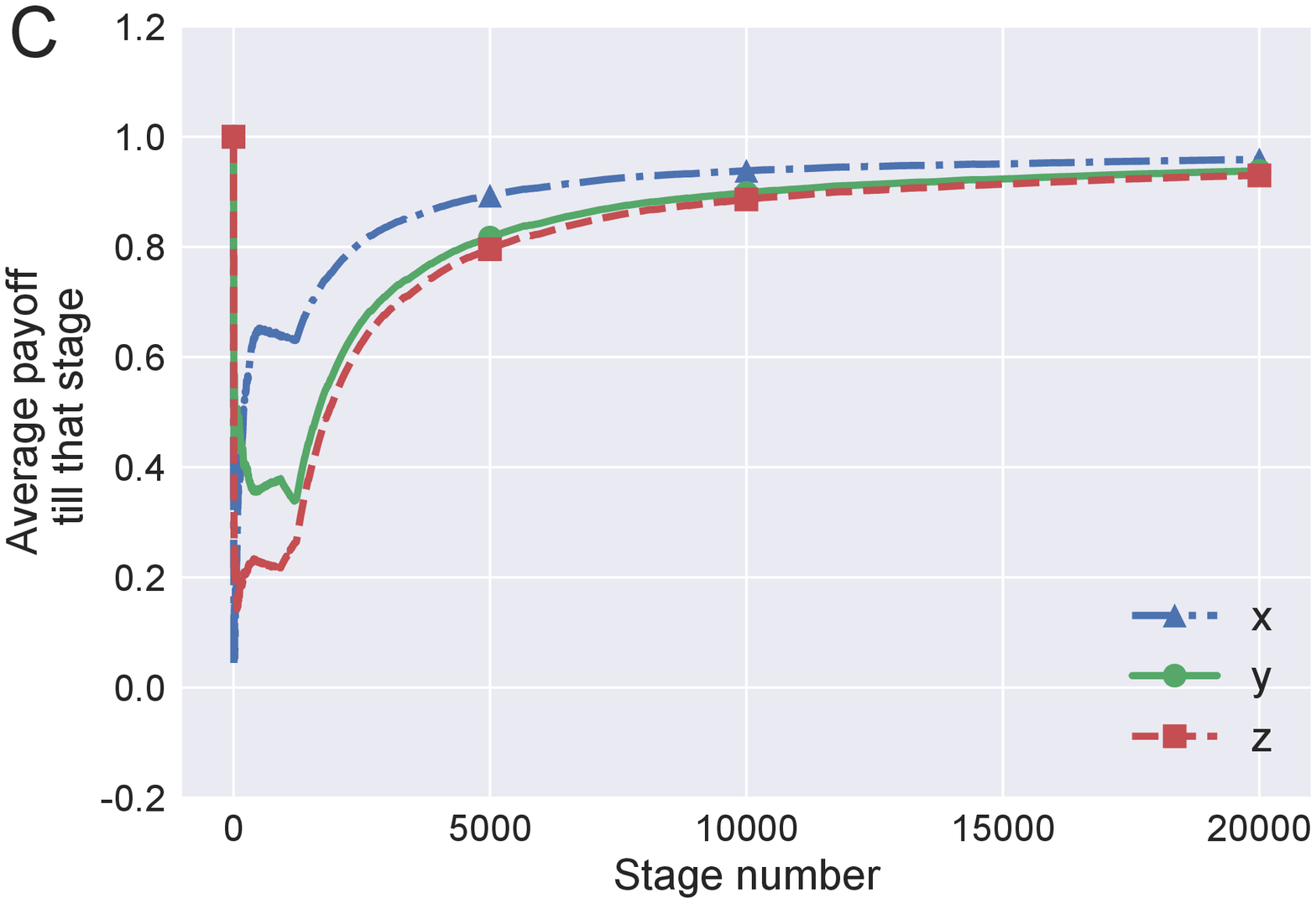}
	\caption{Illustration of average payoffs during the 3-player repeated public goods game (with $r=2$) in three different scenarios.
	The payoff of mutual cooperation is $R_{c, n-1}=1$.
	In each panel, the strategy of player $x$ is fixed to a specific cooperation enforcing strategy \emph{WSLS} whereas the policies of $y$ and $z$ can be \emph{WSLS} or reinforcement learning.
	(A) $y$ uses \emph{WSLS} while $z$ applies reinforcement learning.
	(B) $y$ and $z$ independently evolve their own policies via reinforcement learning.
	(C) In this scenario, the strategy of $x$ is declared beforehand whereas $y$ and $z$ are learning followers who make an alliance.}
	\label{fig: evolution}
\end{figure*}
%

To investigate whether cooperation enforcing strategies can bring about mutual cooperation when their adopters confront learning players, several scenarios need to be analyzed and simulated.
For the first scenario, to analyze the equilibrium, we let all players adopt cooperation enforcing strategies, except for one who uses a reinforcement learning approach.
For the second scenario, there are more independent learning players, i.e., some players use cooperation enforcing strategies while others independently adopt reinforcement learning strategies.
For the third scenario, we assume the cooperation enforcing strategy be declared in advance and analyze its performance when the followers are learning players, which is also known as Stackelberg game \cite{demiguel2009stochastic}.
The simulation results are shown in Figure \ref{fig: evolution}.

In Proposition \ref{proposition: MPE}, we have known that when all players adopt cooperation enforcing strategies, they are in an equilibrium and a player who deviates cannot obtain a payoff greater than that of mutual cooperation.
There remains a question here: if the deviator is not malicious but just a learning player who has no idea of the cooperation enforcing strategies, can she eventually learn to cooperate?
The answer should be yes.
If the strategies of all the other players are fixed to cooperation enforcing strategies, the environment for the only learning player is not that complicated.
Besides, mutual cooperation is the only optimal solution to this scenario.
Therefore, as long as the reinforcement learning player sufficiently explores this environment, she will generate a ``near-optimal'' solution to eq. (\ref{eq: Q}) and evolve to cooperation.
Figure \ref{fig: evolution}(A) illustrates that a learning player learns to cooperate quickly and shares the optimum with all the cooperation enforcing strategy players.

To analyze the case of more complicated scenarios where multiple learning players exist, we set up a simulation in which only a portion of players use cooperation enforcing strategies while others learn.
In this case, however, learning players can not finally find the optimal solution.
This is because they optimize their own policies independently, which will lead to a failure of obtaining the dynamic global knowledge of the whole game system.
If other players adapt their strategies, the environment for one player will be non-stationary and the optimization target for the learning will shift. Actually, multi-agent learning is still a challenging issue in the artificial intelligence community.
Figure \ref{fig: evolution}(B) shows an example of the collapse of mutual cooperation, in which multiple players independently learn their own policies cannot lead to mutual cooperation.

Although it seems difficult for cooperation enforcing strategies to bring about cooperation when their adopters confront multiple independent learning opponents, we can conquer this difficulty by establishing a Stackelberg game.
In this Stackelberg setting, some players need to declare their cooperation enforcing strategies in advance and then the learning players can explore and evolve their strategies.
Eventually, all the learning players will find the best responses to their declared policies, which is to fully cooperate.
We refer to those players who commit strategies as leaders and the others as followers.
Because followers are witting of the strategies of leaders, they can even coordinate their policies with each other and work in collusion to compete with leaders.
We can obviously extend the reinforcement learning method for a single player to this alliance.
The fixed strategies of the leaders produce a stationary environment for the alliance.
Therefore, mutual cooperation can be enforced with a learning alliance.
Figure \ref{fig: evolution}(C) provides an instance that if cooperation enforcing strategies are Stackelberg leader strategies, they do promote cooperation.

\section{Conclusions and Future Work}
We present a class of delicately designed Markov strategies for repeated games.
Using such a strategy, a single player can simultaneously constrain the utilities of all her opponents and enforce mutual cooperation among all of players. The optimal value of any opponent's utility can only be realized through all players' cooperating.
Moreover, we prove that as long as there exists at least one player adopting such a strategy, any type of colluding alliance cannot get a higher payoff than that of $n$-player mutual cooperation.
Our results also show that these strategies can still promote cooperation, even when the opponents are non-strategic learning players. Although we are using the conventional repeated public goods game model, the approaches in our paper could be easily extended to various multi-player games. Immediate future work is about the effect of a non-linear marginal per-capita rate of return on the group cooperation, the generalization of cooperation enforcement in stochastic games with large action space \cite{mcavoy2016autocratic}, as well as the extension into multi-agent systems where agents make use of different orders of theory of mind \cite{albrecht2018autonomous}. Moreover, further research on multi-agent reinforcement learning \cite{perolat2017multi} could be potentially inspired by this work.

%

\bibliographystyle{aaai}
\bibliography{cooperator-aaai19}

\end{document}